\documentclass[letterpaper,11pt]{article}
\usepackage[utf8]{inputenc}

\pdfoutput=1

\usepackage{amsmath, amsthm, amssymb}
\usepackage[linesnumbered,vlined]{algorithm2e}
\usepackage[table,xcdraw]{xcolor}
\usepackage{mathtools}
\usepackage[numbers]{natbib}
\usepackage{comment} 

\usepackage{color-edits} 

\usepackage[most]{tcolorbox}
\usepackage{xfrac}
\usepackage{hyperref}
\usepackage{multirow}
\usepackage{caption}
\usepackage{bm}
\usepackage{newfloat}
\usepackage{enumitem}
\usepackage{xpatch}
\usepackage{soul}
\usepackage{bbm}
\usepackage{tikz}
\usepackage{titlesec}

\usepackage{fancyhdr}

\makeatletter
\newcommand{\thickhline}{%
    \noalign {\ifnum 0=`}\fi \hrule height 1.5pt
    \futurelet \reserved@a \@xhline
}
\makeatother

\usepackage[margin=1in]{geometry}

\allowdisplaybreaks

\definecolor{mygreen}{RGB}{20,120,60}

\title{Stochastic Minimum Vertex Cover in General Graphs:\\
 a $3/2$-Approximation}

\author{ Mahsa Derakhshan\thanks{Northeastern University} \and Naveen Durvasula \thanks{University of California, Berkeley} \and Nika Haghtalab \footnotemark[2]}

\newcommand{\E}[0]{\ensuremath{\mathbb{E}}}

\newcommand{\opt}{\ensuremath{\textsc{opt}}\xspace}
\newcommand{\OPT}{\ensuremath{\textsc{OPT}}\xspace}

\newcommand{\cru}[1]{\ensuremath{G}}
\newcommand{\inm}[1]{\ensuremath{Z}}

\newcommand{\apx}[0]{\ensuremath{\epsilon}}

\newcommand{\pa}[1]{\left(#1\right)}
\newcommand{\bra}[1]{\left[#1\right]}
\newcommand{\set}[1]{\left\{#1\right\}}
\newcommand{\abs}[1]{\left|#1\right|}

\newcommand{\mvc}[1]{\text{MVC}\left(#1 \right)}
\SetKwRepeat{Do}{do}{while}

\renewcommand{\epsilon}[0]{\ensuremath{\varepsilon}}

\let\originalleft\left
\let\originalright\right
\renewcommand{\left}{\mathopen{}\mathclose\bgroup\originalleft}
\renewcommand{\right}{\aftergroup\egroup\originalright}

\addauthor{md}{red}    
\addauthor{nh}{magenta} 
\addauthor{nd}{green} 

\newtheorem{theorem}{Theorem}[section]
\newtheorem{question}{Question}

\newtheorem{lemma}[theorem]{Lemma}

\newtheorem{definition}[theorem]{Definition}
\newtheorem{claim}[theorem]{Claim}

\newtheorem{observation}[theorem]{Observation}
\newtheorem{remark}[theorem]{Remark}

\makeatletter
\def\thm@space@setup{%
  \thm@preskip= 0.2cm
  \thm@postskip=\thm@preskip 
}
\makeatother

\definecolor{mygreen}{RGB}{20,155,20}
\definecolor{myred}{RGB}{195,20,20}
\definecolor{linkcolor}{RGB}{0,0,230}
\definecolor{mylightgray}{RGB}{230,230,230}
\definecolor{verylightgray}{RGB}{240,240,240}
\definecolor{commentcolor}{RGB}{120,120,120}

\SetCommentSty{mycommfont}

\hypersetup{
     colorlinks=true,
     citecolor= mygreen,
     linkcolor= myred,
     urlcolor= mygreen
}

\newcommand{\mc}[1]{\ensuremath{\mathcal{#1}}}

\newcounter{myalgctr}

\newenvironment{tbox}{
\par\addvspace{0.2cm}
\begin{tcolorbox}[width=\textwidth,
                  enhanced,
                  boxsep=2pt,
                  left=1pt,
                  right=1pt,
                  top=4pt,
                  boxrule=1pt,
                  arc=0pt,
                  colback=white,
                  colframe=black,
                  unbreakable
                  ]
}{
\end{tcolorbox}
}

\newenvironment{tboxh}{
\par\addvspace{0.2cm}
\begin{tcolorbox}[width=\textwidth,
                  enhanced,
                  boxsep=2pt,
                  left=1pt,
                  right=1pt,
                  top=4pt,
                  boxrule=1pt,
                  arc=0pt,
                  colback=white,
                  colframe=black,
                  unbreakable,
                  float=t
                  ]
}{
\end{tcolorbox}
}

\newcommand{\tboxhrule}[0]{\vspace{0.1cm} \hrule \vspace{0.2cm}}

\newenvironment{titledtbox}[1]{\begin{tbox}#1 \tboxhrule}{\end{tbox}}
\newenvironment{titledtboxh}[1]{\begin{tboxh}#1 \tboxhrule}{\end{tboxh}}

\newenvironment{tboxalg2e}[1]{
\refstepcounter{myalgctr}
	\begin{titledtbox}{\textbf{Algorithm \themyalgctr.} #1}
	\vspace{-0.2cm}
}
{
	\vspace{-0.3cm}
	\end{titledtbox}
}

\begin{document}
\setlength{\parskip}{0.2em}

\maketitle
\thispagestyle{empty}


\begin{abstract}
We study the stochastic vertex cover problem. In this problem, $G = (V, E)$ is an arbitrary known graph and $\mathcal{G}^\star$ is an unknown random subgraph of $G$ where each edge $e$ is realized independently with probability $p$. Edges of $\mathcal{G}^\star$ can only be verified using edge queries. The goal in this problem is to find a minimum vertex cover of $\mathcal{G}^\star$ using a small number of queries.

\smallskip
Our main result is designing an algorithm that returns a vertex cover of $\mathcal{G}^\star$ with size at most $(3/2+\epsilon)$ times the expected size of the minimum vertex cover, using only $O(n/\epsilon p)$ non-adaptive queries. This improves over the best-known 2-approximation algorithm  by Behnezhad, Blum and Derakhshan [SODA'22] who also show that $\Omega(n/p)$ queries are necessary to achieve any constant approximation. 

\smallskip
Our guarantees also extend to instances where edge realizations are not fully independent. We complement this upperbound with a tight $3/2$-approximation lower bound for stochastic graphs whose edges realizations demonstrate mild correlations. 

\end{abstract}

{

\hypersetup{
     linkcolor= black
}

\clearpage
}

\pagenumbering{arabic}

\section{Introduction}
In the \emph{stochastic vertex cover} problem, we are given an arbitrary base graph $G = (V, E)$ with $n$ vertices but  we do
not know which edges in $E$ actually exist. Rather each edge $e\in E$ is realized independently
with a given existence probability $p_e \in (0,1]$, forming a subgraph $\mathcal{G}^\star$. Our goal is to find a minimum vertex cover of $\mathcal{G}^\star$.
While $\mathcal{G}^\star$ is unknown, one can verify its edge set by querying edges $e \in E$. 
Of interest, then, are algorithms that \emph{query a small subset of edges} and, based on the outcome of these queries, find a \emph{near-optimal vertex cover} of $\mathcal{G}^\star$.
How small should the set of queried edged 
be?
The gold standard in these problems is to non-adaptively issue a number of queries that is linear in  the number of vertices and polynomial in inverse probability 
$p = \min_{e\in E} p_e$. 

While these stochastic settings are primarily concerned with information theoretical questions, most positive results have  focused on problems whose 
non-stochastic counterparts admit computationally-efficient algorithms. Instances include minimum spanning tree~\cite{DBLP:conf/soda/GoemansV04, DBLP:journals/rsa/GoemansV06}, all pairs shortest paths~\cite{DBLP:journals/rsa/Vondrak07}, maximum matching~\cite{blumetal,blumetalOR,AKL16,AKL17,BR18,YM18,sagt19,sosa19,soda19,stoc20,focs20}, $2$-approximate minimum vertex cover, and bipartite minimum vertex cover~\cite{behnezhad2022stochastic}. This emphasis on efficiently solvable problems is not accidental. By and large, structurally-simple properties and heuristics that had long played a key role in understanding and designing computationally efficient algorithms have been used to guide an algorithm in its choice of queries, 
e.g., the Tutte-Berge witness sets~\cite{AKL16}, short augmenting paths~\cite{blumetal}, local computation~\cite{stoc20}, and greedy heuristics~\cite{behnezhad2022stochastic}.

On the other hand, for vertex cover beyond a $2$-approximation, and other computationally hard regimes, lack of 
structurally-simple properties has been a barrier towards solving the stochastic variants of the problems. A natural question here is whether it is possible to obtain any positive results for problems that lack these structure? In this work, we consider this question for the minimum vertex cover problem.

\begin{question}
Can we achieve a better than $2$-approximation for the stochastic minimum vertex cover problem despite the lack of computationally efficient algorithms for this problem?
\end{question}

Our paper answers this question in the affirmative. At a high level, we introduce an algorithm that returns a  vertex cover of $\mathcal{G}^\star$ with probability $1$ whose expected size is at most $3/2+\epsilon$ times that of the minimum vertex cover of $\mathcal{G}^\star$, for any desirably small $\epsilon$. 
The following theorem, which is formally stated in Section~\ref{sec:analysis}, presents our main result.

\begin{theorem}[Upper-bound] \label{thm:upper-mvc}
For any $\epsilon\in (0,0.1)$, there is an algorithm (namely Algorithm~\ref{alg1}) that returns a $\left( 3/2+\epsilon \right)$-approximate solution for the stochastic minimum vertex cover problem
using $O\left(n/\epsilon p\right)$ queries.

\end{theorem} 

The number of queries used in this algorithm is asymptotically optimal since \cite{behnezhad2022stochastic} show that $\Omega(n/p)$ queries are necessary to achieve any constant approximation ratio.
Interestingly, the approximation guarantees of Theorem~\ref{thm:upper-mvc} continue to hold even if edges of $\mathcal{G}^\star$ are correlated (see Section~\ref{sec:lower}.)
This allows us to handle \emph{mild correlations} in the realization of $\mathcal{G}^\star$.  That is, even if  $O(n)$ edges are allowed to be realized in a correlated way\footnote{For a formal definition of a stochastic process with $O(n)$ correlated edges, see Definition~\ref{def:mildcorr}}, we can still achieve  a $\left( 3/2+\epsilon \right)$-approximate solution using only  $O\left(n/\epsilon p\right)$ queries. 
Our next result shows that for such mildly correlated processes, a $3/2$-approximation is the best one can hope to get when using  $O\left(n/\epsilon p\right)$ queries. 
The following theorem, which is formally stated in Section~\ref{sec:lower}, presents our lower bound.

\begin{theorem}[Informal Lower Bound] \label{thm:lower-mvc}
There is a stochastic process for generating $\mathcal{G}^\star$ with $O(n)$ correlated edges, such that any algorithm that returns a vertex cover of $\mathcal{G}^\star$ using only $O(n/\epsilon p)$ queries, must have an approximation ratio of at least $(3/2 - \epsilon)$ with probability $1 - o(1)$.
\end{theorem}

Theorems~\ref{thm:upper-mvc} and \ref{thm:lower-mvc} together demonstrate that our results are tight if there are mild correlations between the edges of $\cal{G}^\star$, which can readily exist in practical applications. Moreover, this shows that to further go beyond the $3/2$-approximation one must fully leverage independence across all edges. Indeed, a similar characterization was given for the stochastic matching problem by~\cite{AKL16} and~\cite{sosa19}.

\paragraph{Mild Correlations and Independence in Stochastic Optimization.}
Correlated realizations have been considered in several stochastic combinatorial optimization problems~\cite{agrawal2012price,ahmed2013probabilistic,AKL16,sosa19,blumetalOR,gupta2011approximation}.
For the stochastic matching problem, mildly correlated graphs were first considered by~\cite{AKL16} who, in addition to their algorithmic results,  provide a construction of stochastic graphs with only $O(n)$ correlated edges which does not admit better than $2/3$-approximation. 
Our lower-bound too uses a similar construction with only $O(n)$ correlated edges. 
Later,~\cite{sosa19} use an elegant matching sparsifier of~\cite{DBLP:conf/icalp/BernsteinS15}  to design a $2/3$-approximation algorithm closing the gap for mildly correlated graphs.
The tight characterization of what is possible for mildy correlated graphs in the stochastic matching problem also paved the way for obtaining a
 $(1-\epsilon)$ approximation ratio for the fully independent setting by~\cite{stoc20}. To fully leverage the independence across all edges,~\cite{stoc20} crucially utilize $(1-\epsilon)$-approximate matching algorithms 
designed in the \textsf{LOCAL} model of computation.
This subsequently resulted in a $(1+\epsilon)-$approximate stochastic vertex cover for bipartite graphs~\cite{behnezhad2022stochastic}.
However, the absence of better than 2-approximation algorithms for minimum vertex cover in the \textsf{LOCAL} model may be seen as an obstacle in breaking the $3/2$ barrier for the stochastic vertex cover on fully independent graphs.

\paragraph{Algorithm Design Overview.}
To achieve a $3/2$-approximation, we start with two approaches to solving the stochastic minimum vertex cover problem that give $2$-approximations. 
Our final algorithm is the result of carefully combining insights from these two approaches.

All of the algorithms in this work follow the same blueprint: We consider a set $P\subseteq V$ and the induced subgraph $H = G[V\setminus P]$. We then query all the edges of $H$ to realize $\mathcal{H}^\star$ and take its vertex cover $M$. Our algorithm then returns $S=P\cup M$. We note that $S$ is a vertex cover of $\mathcal{G^\star}$, since any edge of $\mathcal{G}^\star$ that is not covered by $P$ is covered by $M$. Our algorithms and their guarantees only differ in their choice of $P$.

We give two $2$-approximation algorithms that employ different principles in their choice of $P$.
Algorithm~\ref{alg0} hallucinates a random subgraph of $G$, namely $\mathcal{G}_1$, and uses $P$ that is a minimum vertex cover of $\mathcal{G}_1$.
On the other hand, Algorithm~\ref{alg0-2} estimates the probability that any vertex $v\in V$ would belong to the minimum vertex cover of $\mathcal{G}^\star$, denoted by $c_v$, and uses
$P = \{v\mid c_v> 1/2\}$.
While both of these algorithms achieve a $2$-approximation in the worst-case, their performance guarantees differs based on the distribution of $c_v$'s. In particular, both of these algorithms over-include some vertices --- i.e., include a vertex that does not belong to the minimum vertex cover --- but they differ in the type of vertices they over-include. 
As our analysis shows, the first algorithm significantly  over-includes vertices that have a very small $c_v$, but the second algorithm only over-includes vertices with $c_v > \frac 12$.

Our $3/2$-approximation algorithm (Algorithm~\ref{alg1}) combines these two insights to define the set  $P= P_1 \cup P_2$.
It first chooses $\tau$ that carefully balances the contribution of vertices with $c_v>\tau$ to the expected size of the minimum vertex cover. 
For vertices whose $c_v\in [1-\tau - \epsilon, \tau]$, we use the style of Algorithm~\ref{alg0} and only include them in $P_1$ if they also belong to a minimum vertex cover of a hallucinated random subgraph.
For the set of vertices with $c_v > \tau$, we use the style of Algorithm~\ref{alg0-2} and include all of them in $P_2$.
Our analysis carefully balances out the over-inclusion of vertices  to achieve a $3/2$-approximation. 

As presented above, our algorithm requires the knowledge of $c_v$s , i.e., the probability that a vertex belongs to the minimum vertex cover. However, our algorithms and analysis extend immediately to use estimated values of $c_v$, which can be calculated efficiently when given an oracle for the minimum vertex cover problem. Moreover, our approach can readily work with significant mis-estimates of $c_v$, e.g., it  achieves $3/2+\epsilon$ times the
 approximation factor of any 
 oracle for the minimum vertex cover problem.
We discuss this further in Section~\ref{section:call}.

\section{Notation}


We work with a known arbitrary graph $G = (V, E)$ and 
existence probability $p_e\in(0,1]$ for each $e\in E$.
We consider a random subgraph $\mathcal{G}^\star$ in which every edge $e\in E$ is realized with probability $p_e$, independently. We denote $p = \min_{e\in E} p_e$.

Let $\text{MVC}$ denote a function that given any input graph outputs a minimum vertex cover of that graph. We may also refer to this as the \emph{minimum vertex cover} oracle. 
We define $\OPT=\text{MVC}(\mathcal{G}^\star)$ to be the optimal solution of our problem. Note that \OPT is a random variable since $\mathcal{G}^\star$ itself is a random realization of $G$.
We also let $\opt=\E[|\OPT|]$ be the expected size of this optimal solution. Moreover, for any vertex $v\in V$, we define $$c_v=\Pr[v\in \OPT],$$ which is the probability that $v$ joins the optimal solution. This implies $\sum_{v\in V} c_v = \opt.$
Similarly, for any edge $e=(u,v)$ in graph $G$, we let $c_e$ be the probability that this edge is covered by $\OPT$. That is, 
$$c_{(u,v)}=\Pr[u\in \OPT \text{ or } v\in \OPT].$$
When $c_v$s and $c_e$s are not known in advance, we use a polynomial number of calls to a minimum vertex cover oracle to estimate them within arbitrary accuracy. See Section~\ref{section:call} for more details regarding these estimates.

\section{Warm Up -- Beating 2-Approximation}

In this section, we start by discussing two simplified variants of our $3/2$-approximate algorithm. Both of these algorithms have the worst case approximation ratio of $2$. However, their performance varies for different instances of the problem. One of them has a better performance if a large portion of $\OPT$ comes from vertices with smaller $c_v$'s while the other one prefers a large portion of $\OPT$ to be from vertices with larger $c_v$'s. 
After discussing these two algorithms, we will show how, due to their opposing nature, running the best of the two algorithms beats the 2-approximation ratio. 
Finally, we explore how this observation inspires the design of our $3/2$-approximation algorithm.

All the algorithms we design in this paper follow a similar framework. In all of them, we first pick a subset of vertices $P$ and commit to adding them to the final vertex cover. As a result, we only need to query the edges not covered by these vertices. We denote this subgraph by $H$. Formally, $H = G[V\setminus P]$ is the subgraph induced in $G$ by $V\setminus P$. After querying $H$, we find a vertex cover of its realized edges which we denote by $M$. Finally, we output $P\cup M$.
Our algorithms mainly differ in their choice of $P$. See~Figure~\ref{fig:algdesc} for an illustration of our framework.

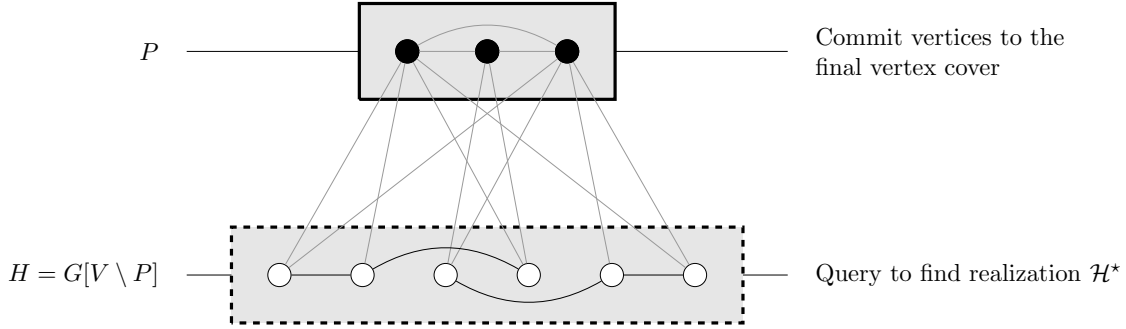
\begin{figure}[h]
    \centering
    \scalebox{0.85}{\begin{tikzpicture}
\fill [color=black!10] (1,1) rectangle (9,2.5);
\fill [color=black!10] (3,4.5) rectangle (7,6);

\draw[dashed, line width = 0.5mm] (1,1) --(1,2.5) -- (9,2.5) -- (9,1) -- (1,1);

\draw[line width = 0.5mm] (3,4.5) --(7,4.5) -- (7,6) -- (3,6) -- (3,4.5);

\node[draw, shape=circle, fill=black] (1p) at (3.75 ,5.25) {};
\node[draw, shape=circle, fill=black] (2p) at (5 ,5.25) {};
\node[draw, shape=circle, fill=black] (3p) at (6.25 ,5.25) {};

\node[draw, shape=circle, fill=black] (1p) at (3.75 ,5.25) {};
\node[draw, shape=circle, fill=black] (2p) at (5 ,5.25) {};
\node[draw, shape=circle, fill=black] (3p) at (6.25 ,5.25) {};

\node[draw, shape=circle, fill=white] (1h) at (1.75 ,1.75) {};
\node[draw, shape=circle, fill=white] (2h) at (3.05 ,1.75) {};
\node[draw, shape=circle, fill=white] (3h) at (4.35 ,1.75) {};
\node[draw, shape=circle, fill=white] (4h) at (5.65 ,1.75) {};
\node[draw, shape=circle, fill=white] (5h) at (6.95 ,1.75) {};
\node[draw, shape=circle, fill=white] (6h) at (8.25,1.75) {};

\draw[color=black!40] (1p) edge[bend left=30] (3p);
\draw[color=black!40] (1p) -- (2p);
\draw[color=black!40] (2p) -- (3p);

\draw[color=black!40] (1p) -- (1h);
\draw[color=black!40] (1p) -- (2h);
\draw[color=black!40] (1p) -- (4h);
\draw[color=black!40] (1p) -- (6h);

\draw[color=black!40] (2p) -- (3h);
\draw[color=black!40] (2p) -- (4h);

\draw[color=black!40] (3p) -- (1h);
\draw[color=black!40] (3p) -- (3h);
\draw[color=black!40] (3p) -- (5h);
\draw[color=black!40] (3p) -- (6h);

\draw (1h) -- (2h);
\draw (2h) edge[bend left=30] (4h);
\draw (3h) edge[bend right=30] (5h);
\draw (5h) -- (6h);

\draw (3,5.25) -- (0.3, 5.25);
\draw (7, 5.25) -- (9.7,5.25);

\node[align=right, anchor=east] (plabel) at (0, 5.25) {$P$};
\node[align=right, anchor=east] (hlabel) at (0, 1.75) {$H = G[V \setminus P]$};

\draw (1,1.75) -- (0.3, 1.75);
\draw (9, 1.75) -- (9.7,1.75);

\node[align=left, anchor=west] (pdesc) at (10, 5.25) {Commit vertices to the\\final vertex cover};
\node[align=left, anchor=west] (hdesc) at (10, 1.75) {Query to find realization $\mathcal H^\star$};
\end{tikzpicture}}
    \caption{\textbf{A commit, then query approach.} After committing vertices of $P$ to the final vertex cover, we query the subgraph of edges not covered by $P$ which we denote by $H$.  This is the subgraph of $G$ induced by $V\setminus P$. The output of our algorithm is $P\cup M$ where $M$ is a vertex cover of $\mathcal{H}^\star$.}
    \label{fig:algdesc}
\end{figure}

We make the following observation about our algorithm framework.
\begin{observation}\label{obs:itisavertexcover}
Let $P\subseteq V$ and $H=G[V\setminus P]$ be the induced subgraph on $V\setminus P$. Let $\mathcal{H}^\star$ be the realization of the edges of $H$  and $M$ be a vertex cover of $\mathcal{H}^\star$.
Then,  $P\cup M$ is a vertex cover of $\mathcal{G}^\star$.
\end{observation}

The simplest way of picking $M$ is for it to be a minimum vertex cover of all the realized edges of $H$ (i.e. $\mathcal{H}^\star$). However, sometimes, for the sake of analysis we require $M$ to be a vertex cover of $H$ satisfying a certain property. 
In particular, we want $\Pr[v\in M]=c_v$. We achieve this by letting $M$ be a minimum vertex cover of all the realized edges of $H$ and a  hallucination of $G\setminus H$. We will explain this in more detail in the following algorithm.

\paragraph{The first 2-approximate algorithm.}  To construct the subset $P$, this algorithm (formally stated as Algorithm~\ref{alg0}) hallucinates a random realization of $G$ and lets $P$ be its minimum vertex cover. Following the aforementioned framework, the next step is to find $M$: a vertex cover of $\mathcal{H}^\star$.
In order to do that, the algorithm hallucinates another realization of $G$ by including any edge $e\in G\setminus H$ with probability $p_e$, independently, and including edges of $\mathcal{H}^\star$.
Finally, it finds an MVC of this realization denoted by $M$, and outputs $P\cup M$. 
  \vspace{2 mm}
 \begin{tboxalg2e}{A $2$-approximation stochastic vertex cover algorithm}
\begin{algorithm}[H]
	\DontPrintSemicolon
	\SetAlgoSkip{bigskip}
	\SetAlgoInsideSkip{}
	\label{alg0}
	Let $\mathcal{G}_{1}$ be a random realization of $G$ containing any edge $e\in G$ independently w.p. $p_e$.\;
		  $P \gets \text{MVC}(\mathcal{G}_{1})$\;
	Let $H$ be the subgraph induced in $G$ by $V\setminus P$.\;
Query subgraph $H$ and let $\mathcal{H}^\star$ be its realization.\;
	Let $\mathcal{G}_2$ be a subgraph of $G$ containing  all the edges in $\mathcal{H}^\star$ and any edge $e\in G\setminus H$ independently w.p. $p_e$.\;
$M\gets \text{MVC}(\mathcal{G}_{2})$\;
	Return $P\cup M$
\end{algorithm}
\end{tboxalg2e}
\vspace{3 mm}

We will first prove that this algorithm queries only  $O(n/p)$ edges, that is $|H|=O(n/p)$. This is due to the fact that, by definition, none of the edges of $H$ are realized in $\mathcal{G}_{1}$ since otherwise, one of its vertices should be in vertex cover $P$ (which is a contradiction). For any subgraph with more than $n/p$ edges the probability of none of its edges being in $\mathcal{G}_{1}$ is at most $(1-p)^{n/p}$. Moreover since $H$ is an induced subgraph of $G$, there are at most $2^n$ possibilities for it. By an application of union bound, we see that w.h.p., $H$ has at most $n/p$ edges. That is:
 $$\Pr[|H|\leq n/p] \geq 1- 2^n(1-p)^{n/p}\geq 1-\frac{2^n}{e^{n}}=1-(2/e)^n. $$ We state this observation below for future reference.

\begin{observation}\label{claim:nkergkjer}
Let $\mathcal{G}$ be a random realization of $G$ containing any edge $e\in G$ independently with probability $p_e$, and let $M$ be a minimum vertex cover of $\mathcal{G}$. The number of edges in $G$ not covered by $M$ is $O(n/p)$.
\end{observation}

As mentioned earlier, it is only for the sake of analysis that we do not simply let $M$ be an arbitrary minimum vertex cover of $\mathcal{H}^\star$. 
Instead, we use $M = \mvc{\mathcal{G}_2}$ since it satisfies that 
$\Pr[v\in M] = c_v.$ 
This is due to the fact that $\mathcal{G}_2$ and $\mathcal{G}^\star$ are drawn from the same distribution, i.e., 
 $\mathcal{G}_2$ contains any edge $e$ independently with probability $p_e$.
 Below, we state this observation formally.  

 \begin{observation}\label{obs:cv}
 Let $H$ be a subgraph of $G$ and  $\mathcal{H}^\star$ be its actual realization. Moreover, we define $\mathcal{G}_2$ to be a subgraph of $G$ containing  all the edges  in $\mathcal{H}^\star$ and any edge $e\in G\setminus H$ independently with probability $p_e$.  If $M$ is a minimum vertex cover of  $\mathcal{G}_2$ it satisfies $\Pr[v\in M]= c_v.$
 \end{observation}

As a result of this observation and the fact that $P$ also comes from the same distribution as $\OPT$, we get $\E[|P\cup M|]\leq 2\opt$ which implies that Algorithm~\ref{alg0} is a 2-approximation.  However, we claim that depending on the the way $c_v$'s are distributed this algorithm may result in a better than 2 approximation ratio. First of all, observe that $\mathcal{G}_1$ and $\mathcal{G}_2$ are two independent random variables. The only way in which $\mathcal{G}_1$ impacts the construction of $\mathcal{G}_2$ is in determining which of its edges come from the actual realization and which ones are hallucinated. Nonetheless, $\mathcal{G}_2$ contains any edge $e$ with probability $p_e$ independently from other edges and from $\mathcal{G}_1$. This implies that $P=\mvc{\mathcal{G}_1}$ and $M = \mvc{\mathcal{G}_2}$ are also two independent random variables. Therefore, any vertex $v$ joins $P$ and $M$ independently with probability $c_v$ which means 
\begin{align}\label{eq:PM}
    \Pr[v\in P\cup M] = 1-(1-c_v)^2 =  c_v(2-c_v).
\end{align} 
Since $c_v(2-c_v) \approx 2c_v$ for $c_v\rightarrow 0$,
the worst case scenario for this algorithm is when all vertices have very small $c_v$s. On the other hand, if for all the vertices we have $c_v\geq 0.5$ then, $c_v(2-c_v)\leq 1.5c_v$ which results in an approximation ratio of $3/2$ (our desired bound). Having established this, we will next discuss another 2-approximation algorithm with an opposite nature. This algorithm has a better performance if a large portion of the optimal solution comes from vertices with small $c_v$s.

\paragraph{The second $2$-approximate algorithm.} This algorithm (formally stated as Algorithm~\ref{alg0-2}) is not exactly a 2-approximation. Instead, it finds a $(2+O(\epsilon))$-approximate vertex cover using $O(n/\epsilon p)$ queries for any $\epsilon\in (0,0.1).$ In this algorithm, we set $P=\{v\in V: c_v\geq 0.5-\epsilon\}$. The rest of the algorithm follows our standard framework similar to Algorithm~\ref{alg0}.
\vspace{2 mm}
 \begin{tboxalg2e}{A $(2+O(\epsilon))$-approximation stochastic vertex cover algorithm}
\begin{algorithm}[H]
	\DontPrintSemicolon
	\SetAlgoSkip{bigskip}
	\SetAlgoInsideSkip{}
	\label{alg0-2}
	Define $P= \{v\in V :  c_v\geq 0.5-\epsilon\}$.\;
	Let $H$ be the subgraph induced in $G$ by $V\setminus P$.\;
Query subgraph $H$ and let $\mathcal{H}^\star$ be its realization.\;
	Let $\mathcal{G}$ be a subgraph of $G$ containing  all the edges in $\mathcal{H}^\star$ and any edge $e\in G\setminus H$ independently w.p. $p_e$.\;
$M\gets \text{MVC}(\mathcal{G})$\;
	Return $P\cup M$
\end{algorithm}
\end{tboxalg2e}
\vspace{3 mm}

We will first discuss why the above mentioned algorithm queries only $O(n/\epsilon p)$ edges. To put Observation~\ref{claim:nkergkjer} differently, the expected number of edges in $G$ not covered by $\OPT$, or equivalently $\sum_e (1-c_e)$, is $O(n/p)$. Using Markov's inequality, 
this implies that the number of edges with $c_e < 1-\epsilon$ is upper-bounded by $O(n/\epsilon p)$. Since both end-points of any edge $e\in H$ have $c_v<0.5-\epsilon$, all these edges have $c_e \leq 1-2\epsilon$. Hence, there are at most $O(n/\epsilon p)$ of them.

We will next prove that Algorithm~\ref{alg0-2} is  a $(2+O(\epsilon))$-approximation. Clearly, for any vertex $v\in P$ we have $\Pr[v\in P\cup M] = 1$. Moreover, any $v\notin P$ joins $M$ with probability $c_v$ (by Observation~\ref{obs:cv}) which  implies $\Pr[v\in P\cup M] = c_v$. As a result
\begin{align}\label{eq:ilwflk}
    \E[|P\cup M|]= \sum_{v\in P} \Pr[v\in P\cup M]+ \sum_{v\notin P} \Pr[v\in P\cup M] = 
    \E[|P|]+ \sum_{v\notin P} c_v.
\end{align}
Let us define $\alpha$ to be the fraction of the optimal solution not in $P$.
That is  
\begin{align*}
    \alpha = \frac{\sum_{v\notin P} c_v}{\opt}.
\end{align*}
Since, by definition, any vertex $v\in P$ satisfies $c_v\geq 0.5-\epsilon$ we have
$$\E[|P|]\times(0.5-\epsilon) \leq \sum_{v\in P} c_v = \opt - \sum_{v\notin P} c_v =(1-\alpha)\opt,$$ which gives us 
$$\E[|P|]\leq \frac{(1-\alpha)\opt}{0.5-\epsilon}.$$
Combining this with \eqref{eq:ilwflk} gives us

\begin{align}\label{eq:s2}
  \E[|P\cup M|= \E[|P|]+ \sum_{v\notin P} c_v \leq \frac{(1-\alpha)\opt}{0.5-\epsilon}+ \alpha \opt = \opt\frac{1-\alpha(0.5+\epsilon)}{0.5-\epsilon} = \opt(2-\alpha+O(\epsilon)),
\end{align} and implies the $(2+O(\epsilon))$-approximation ratio. Observe that this bound is tight only when $\alpha =0$. If for an instance of the problem, a large number of vertices have $c_v<0.5-\epsilon$ and as a result $\alpha$ is large, this algorithm achieves a better approximation ratio.

\subsection{Beating  $2$-approximation}
As discussed above, Algorithm~\ref{alg0-2} has a better performance when a large portion of $\OPT$ comes from vertices with $c_v < 0.5-\epsilon$ while Algorithm~\ref{alg0} is almost the opposite. Therefore, an idea for beating the 2-approximation ratio is to run the best of these two. We will prove that doing so achieves an approximation ratio of $5/3 + O(\epsilon)$. Given parameter $\epsilon\in (0,0.1)$, let us recall the definition of $\alpha$ as  \begin{align}\label{eq:alpha}
    \alpha = \frac{\sum_{v: c_v< 0.5-\epsilon} c_v}{\opt}.
\end{align} Moreover, let $S_1$ and $S_2$ be the solutions outputted by Algorithm~\ref{alg0} and Algorithm~\ref{alg0-2} respectively. As an upper-bound for $|S_1|$ we have  
\begin{align*}
   \E[|S_1|]&=\sum_{v\in V}\Pr[v\in S_1] \stackrel{\eqref{eq:PM}}{=} \sum_{v\in V}c_v(2-c_v) \leq \sum_{\mathclap{v: c_v\geq 0.5-\epsilon}} (1.5+\epsilon)\cdot c_v + \sum_{\mathclap{v: c_v< 0.5-\epsilon}} 2\cdot c_v \\
    & \stackrel{\eqref{eq:alpha}}{=}  (1.5+\epsilon)  (1- \alpha)\opt  + 2\alpha\opt = \opt(1.5 + 0.5\alpha + O(\epsilon)).
\end{align*}
Moreover, by \eqref{eq:s2} we have 
$$\E[|S_2|]= \opt(2-\alpha+ O(\epsilon))$$
Therefore, the approximation ratio achieved by running the best of these two algorithms is upper-bounded by 
\begin{align*}
    \max\big[(2-\alpha+ O(\epsilon)), (1.5 + 0.5\alpha + O(\epsilon))\big]
\end{align*}
We observe that this term is minimized for $\alpha=1/3$ which results in an approximation ratio of $5/3+O(\epsilon).$ This analysis is tight since both algorithms achieve this approximation ratio when $2/3$ of $\OPT$ comes from vertices with $c_v=0.5$ and the rest from vertices with $c_v\rightarrow 0$.

\section{The $(3/2+\epsilon)$-Approximation Algorithm}

Inspired by the above algorithms, in this section, we design a $3/2$-approximation algorithm. 
{Throughout this section, we assume that $c_v$s are known in advance. We relax this assumption in Section~\ref{section:call} by directly estimating $c_v$s to an arbitrary desirable accuracy with polynomial number of calls to a minimum vertex cover oracle.}

Similar to Algorithm~\ref{alg0} and Algorithm~\ref{alg0-2}, this algorithm first picks a subset of vertices $P$ and commits to including them in the final solution and then queries the edges not covered by them, {i.e., $H = G[V\setminus P]$.} The algorithm first picks a threshold $\tau$ which may vary for different instances. Based on this threshold and $c_v$ of the vertices it 
{commits to including the set $P= P_1\cup P_2$.
For vertices whose $c_v\in [1-\tau - \epsilon, \tau]$, we use the style of Algorithm~\ref{alg0} and only include them in $P_1$ if they also belong to a minimum vertex cover of a hallucinated random subgraph.
For the set of vertices with $c_v > \tau$, we use the style of Algorithm~\ref{alg0-2} and include all of them in $P_2$.}

Consider a solution $S$ for a given instance of the problem. If  we claim that $S$ is an $(3/2+\epsilon)$-approximate solution, we need to show
$$\sum_{v\in V}\Pr[v\in S] \leq (3/2+\epsilon)\cdot\opt =  (3/2+\epsilon)\sum_{v\in V} c_v.$$
In other words, $S$ should satisfy
$$\sum_{v\in V}\big((3/2+\epsilon)\cdot c_v - \Pr[v\in S]\big)\geq 0.$$
Inspired by this, for any vertex $v$, we define the budget of this vertex as
\begin{align}\label{eq:budget}
    b_v= \max\Big((3/2+\epsilon)\cdot c_v - \Pr[v\in S],\, 0\Big)
\end{align}
and its cost as 
\begin{align}\label{eq:cost}
\sigma_v= \max\Big(\Pr[v\in S]- (3/2+\epsilon)\cdot c_v,\, 0\Big)    
\end{align}
Proving that $S$ is an $3/2$-approximate solution is equivalent to showing that we can use the budget of vertices with $b_v>0$, to pay the cost of the vertices with $\sigma_v>0$.  Formally,
\begin{claim}\label{claim:costbudget}
Let $S$ be a vertex cover of $\mathcal{G}^{\star}$. Also, for any vertex $v\in V$, consider $b_v$ and $\sigma_v$ defined respectively in \eqref{eq:budget} and \eqref{eq:cost}. If $S$ satisfies 
$$\sum_{v\in V} b_v - \sum_{v\in V} \sigma_v \geq 0,$$ then 
$\E[|S|]\leq (3/2+\epsilon)\opt.$
\end{claim}
\begin{proof}
Since for any vertex $v\in V$, we have \begin{align*}
    b_v-\sigma_v &=  \max\Big((1.5+\epsilon)\cdot c_v - \Pr[v\in S],\, 0\Big)- \max\Big(\Pr[v\in S]- (1.5+\epsilon)\cdot c_v,\, 0\Big)\\
    &= (1.5+\epsilon)\cdot c_v - \Pr[v\in S],
\end{align*}
We get 
\begin{align*}
    \sum_{v\in V}  (b_v-\sigma_v)  &= \sum_{v\in V}\Big((1.5+\epsilon)\cdot c_v - \Pr[v\in S]\Big) = \sum_{v\in V}(1.5+\epsilon)\cdot c_v -\sum_{v\in V} \Pr[v\in S]\\
    &= (1.5+\epsilon)\opt - \E[|S|].
\end{align*} 
Therefore, inequality $\sum_{v\in V}  (b_v-\sigma_v)\geq 0$ in the statement of this claim also implies $$(1.5+\epsilon)\opt - \E[|S|]\geq 0,$$ and as a result we have $\E[|S|]\leq (1.5+\epsilon)\opt,$ completing the proof of this claim.
\end{proof}
We observe that if $S$ is found by Algorithm~\ref{alg0}, then vertices with $c_v< (0.5-\epsilon)$ have a positive cost. On the other hand, if $S$ is found by Algorithm~\ref{alg0-2}, all these vertices have a positive budget. Based on this observation, we want a threshold $\tau$ such that:

\begin{itemize}
    \item If Algorithm~\ref{alg0} is run on $\{ v\in V: c_v\in [1-\tau-\epsilon, \tau]\}$, then we can pay the cost of vertices with $c_v< (0.5-\epsilon)$ in this set using the budget of the ones with $c_v> (0.5-\epsilon)$.
    \item If Algorithm~\ref{alg0-2} is run on the rest of the vertices, we can use the budget of the vertices with $c_v< 1-\tau-\epsilon$ to pay the cost of the vertices with $c_v > \tau$.
\end{itemize}
We claim that setting $\tau$ to be the smallest number in $[0.5,1]$ satisfying $$\sum_{\mathclap{v: c_v>\tau}} c_v \leq \,\,\,\, \sum_{\mathclap{v: c_v<1-\tau-\apx}} c_v$$ gives us these properties. However, clearly, we cannot just run two separate algorithms on these two subsets of vertices since there can potentially be a large number of edges between them. Therefore, we need to prove that following this intuition does not force us to query a large number of edges. We formally state our $3/2$-approximate algorithm below as Algorithm~\ref{alg1}. Later, in Section~\ref{sec:analysis}, we prove that for any $\epsilon\in (0,0.1)$ this algorithm outputs a $(3/2+\epsilon)$-approximate vertex cover using only $O(n/\epsilon p)$ queries.

\vspace{3mm}
\begin{tboxalg2e}{Our $3/2$-approximation algorithm.}
\begin{algorithm}[H]
	\DontPrintSemicolon
	\SetAlgoSkip{bigskip}
	\SetAlgoInsideSkip{}
	\label{alg1}
	Let $\mathcal{G}_1$ be a random realization of $G$ containing any edge $e\in G$ independently with probability $p_e$.\;
		 $C\gets \text{MVC}(\mathcal{G}_1).$\;
	 Let $\tau$ be the smallest number in $[0.5,1]$ such that $\sum_{v: c_v>\tau} c_v \leq \sum_{v: c_v<1-\tau-\apx} c_v$.\;
	$P\gets\{v\in V: c_v>\tau\} \cup \{v\in V: c_v\in[1-\tau-\apx, \tau] \text{ and } v\in C\}$.\;
	Let $H$ be the subgraph induced in $G$ by $V\setminus P$.\;
	Query edges in $H$ and let $\mathcal{H}^\star$ be its realization\;
	\Return $P\cup \text{MVC}(\mathcal{H}^\star)$.
\end{algorithm}
\end{tboxalg2e}
\vspace{3 mm}

First, since this algorithm follows our standard framework, by Observation~\ref{obs:itisavertexcover} it outputs a vertex cover of $\mathcal{G}^\star$. Note that in this algorithm, any vertex in set  $\{ v\in V: c_v\in [1-\tau-\epsilon, \tau]\}$ joins set $P$ iff it is in vertex cover $C$, This is similar to the way Algorithm~\ref{alg0} constructs $P$. Therefore for any vertex $v$ in this set, the probability of $v$ joining $P$ in Algorithm~\ref{alg1} is the same as that of Algorithm~\ref{alg0}. On the other hand, from the rest of the vertices, i.e., $\{ v\in V: c_v < 1-\tau-\epsilon \text{ or } \tau<c_v]$, $P$ includes any vertex with $\tau<c_v$. These are the only vertices in this set with $c_v\geq 0.5-\epsilon$. Therefore, this is similar to the way Algorithm~\ref{alg0-2} constructs set $P$.

\section{The Analysis}\label{sec:analysis}
In the following lemma we prove that the number of edges queried by our algorithm is $O(n/\epsilon p).$ 
\begin{lemma}\label{lemma:query}
Subgraph $H$ from Algorithm~\ref{alg1} satisfies $\E[|H|]=O(\frac{n}{\apx p}).$
\end{lemma}
\begin{proof}
Consider an edge $e=(u,v)$. We will first show that if $e\in H$, then either $c_e\leq 1-\apx$ or it is not covered in $C$.  Assume w.l.o.g. that $c_v\geq c_u$. Since $e\in H$, we know that $v\notin P$ and $u\notin P$ which implies $c_v\leq \tau$ since otherwise $v$ joins $P$. We will prove our claim by considering all possible values of $c_u$.
\begin{itemize}
\item $c_u\geq 1-\tau-\apx$: Since we know $c_v\geq c_u$ and $c_v\leq \tau$, in this case both $c_u$ and $c_v$ are in $[1-\tau-\apx, \tau]$. Thus, $e$ joins $H$ iff both its end-points are not in $C$ which means $e$ is not covered in $C$.
\item $c_u< 1-\tau-\apx$: Since we know $c_v\leq \tau$, this implies $c_v+c_u\leq 1-\apx$. Moreover, since for any edge $c_v+c_u\geq c_e$ holds, we get $c_e\leq 1-\apx$.
\end{itemize}
As a result, we have 
$$|H|\leq |\{e: c_e\leq 1-\apx\}| + |\{e: e \text{ not covered by } C\}|.$$
Observation~\ref{claim:nkergkjer}, states that the number of edges not covered by an MVC of a random realization of $G$ is $O(n/p)$. This directly implies 
$$\E[|\{e: e \text{ not covered by } C\}|] = O(n/p).$$
Since $\OPT$ itself is an MVC of a random realization of $G$,  Observation~\ref{claim:nkergkjer} also implies
$$\E |\{e: e \text{ not covered by } \OPT\}|]= \sum_e (1-c_e)= O(n/p).$$
Using Markov's inequality,  this gives us  $$\E[|\{e: c_e\leq 1-\apx\}|] = \E[|\{e: 1-c_e > \apx\}|]\leq \left(\sum_{e\in G} (1-c_e)\right)/\apx = O(n/\apx p).$$
Putting these inequalities together, we conclude 
$$\E[|H|]\leq \E[|\{e: c_e\leq 1-\apx\}| + |\{e: e \text{ not covered by } C\}|] = O(n/\apx p),$$
completing the proof of this claim.	
\end{proof}

As we mentioned before, for the sake of analysis, we need for the vertex cover of $\mathcal{H}^\star$ to include any vertex $v$ with probability $c_v$. In order to achieve this, we need to make a slight change to the algorithm. We explain the modified algorithm below.

\vspace{2 mm}
\begin{tboxalg2e}{An algorithm used only for analysis.}
\begin{algorithm}[H]
    \SetKwFunction{SVC}{SVC}
	\DontPrintSemicolon
	\SetAlgoSkip{bigskip}
	\SetAlgoInsideSkip{}
	\label{alg3}
	\SetKwProg{Fn}{Function}{:}{}
 	\DontPrintSemicolon
Consider subgraph $H$ and subset of vertices $P$ from Algorithm~\ref{alg1}\; 
	Query subgroup $H$ and let $\mathcal{H}^\star$ be its realization.\;
	Let $\bar{H} \gets G\setminus H$\;
	Let $\bar{\mathcal{H}}$ be a random realization of $\bar{H}$ containing each of its edges $e$ independently with probability $p_e$.\;
	$M \gets \text{MVC}(\mathcal{H}^\star \cup \bar{\mathcal{H}})$.\;
   	\Return $M\cup P$.
\end{algorithm}
\end{tboxalg2e}
\vspace{3 mm}

In the rest of the paper we will prove our desired approximation ratio for Algorithm~\ref{alg3} instead of Algorithm~\ref{alg1}. However, for that to work we first need the following observation.

\begin{observation}\label{obs:alg1alg3}
Let $S_1$ and $S_2$ be respectively the outputs of Algorithm~\ref{alg1} and Algorithm~\ref{alg3}. We have $\E[|S_1|]\leq \E[|S_2|]$.
\end{observation}

\begin{proof}
Note that both $S_1$ and $S_2$ contain $P$ and a vertex cover of $\mathcal{H}^\star$. Since the vertex cover of $\mathcal{H}^\star$ in Algorithm~\ref{alg1} is the smallest one, the output of this algorithm is not larger than that of Algorithm~\ref{alg3}.
\end{proof}
We are now ready to prove the main lemma about the size of the solution outputted by Algorithm~\ref{alg3}. The lemma is sated below.
\begin{lemma}\label{lemma:approxratio}
Let $S$ be the the output of Algorithm~\ref{alg3}. We have $$E[|S|]\leq (3/2+\epsilon)\opt.$$
\end{lemma}
\begin{proof}
By Claim~\ref{claim:costbudget}, to prove this lemma it suffices to show $$\sum_{v\in V} b_v - \sum_{v\in V} \sigma_v \geq 0$$ where $b_v$ and $\sigma_v$, the budget and cost of vertex $v$ are respectively defined in \eqref{eq:budget} and \eqref{eq:cost}.
To prove this, we divide the vertices of our graph $V$ to three disjoint subsets $V_1$, $V_2$, and $V_3$ and prove this equation for them separately. That is, for any $V_i$ we prove 
\begin{align}\sum_{v\in V_i} b_v - \sum_{v\in V_i} \sigma_v \geq 0.\label{eq:nkwefjw} \end{align}
 We define these subsets as follows (visialized in Figure~\ref{fig:charging}):
 \begin{itemize}
 \item $V_1=\{v: 0.5-\apx\leq c_v\leq 0.5\}.$	
 \item $V_2=\{v: c_v> \tau\} \cup \{v: c_v<1-\tau-\apx\}.$
 \item $V_3= \{v: 0.5< c_v \leq \tau\}  \cup \{v: 1-\tau-\apx\leq c_v< 0.5-\apx\}.$
 \end{itemize}
    We will prove Equation~\eqref{eq:nkwefjw} for subsets $V_1$, $V_2$ and $V_3$ respectively in Lemma~\ref{lemma:S1}, Lemma~\ref{lemma:S2}, and Lemma~\ref{lemma:S3}. 

Since these three subsets are disjoint and satisfy $V=V_1\cup V_2\cup V_3$, we get
\begin{align}\label{eq:kjfjewjfnjer}
    \sum_{v\in V} b_v - \sum_{v\in V} \sigma_v = \sum_{i} \left(\sum_{v\in V_i} b_v - \sum_{v\in V_i} \sigma_v\right)\geq 0
\end{align} 
completing the proof of this lemma.
\end{proof}

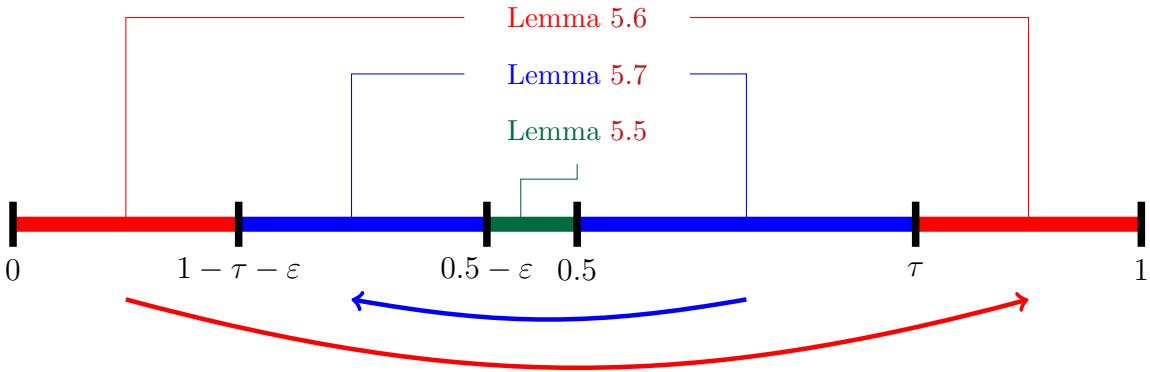
\begin{figure}[h]
    \centering
    \begin{tikzpicture}
\definecolor{cadmium}{rgb}{0.0, 0.42, 0.24}

\draw[color=red, line width = 2mm] (0,0) -- (3, 0);
\draw[color=blue, line width = 2mm] (3,0) -- (6.3, 0);
\draw[color=cadmium, line width = 2mm] (6.3,0) -- (7.5, 0);
\draw[color=blue, line width = 2mm] (7.5,0) -- (12, 0);
\draw[color=red, line width = 2mm] (12,0) -- (15, 0);

\draw[line width = 1mm] (0, 0.3) -- (0,-0.3);
\draw[line width = 1mm] (3, 0.3) -- (3,-0.3);
\draw[line width = 1mm] (6.3, 0.3) -- (6.3,-0.3);
\draw[line width = 1mm] (7.5, 0.3) -- (7.5,-0.3);
\draw[line width = 1mm] (12, 0.3) -- (12,-0.3);
\draw[line width = 1mm] (15, 0.3) -- (15,-0.3);

\node at (0, -0.6) {\large $0$};
\node at (3, -0.6) {\large $1 - \tau - \epsilon$};
\node at (6.3, -0.6) {\large $0.5 - \epsilon$ };
\node at (7.5, -0.6) {\large $0.5$};
\node at (12, -0.6) {\large $\tau$};
\node at (15, -0.6) {\large $1$};

\path[->, color = red, line width = 0.6mm] (1.5, -1) edge[bend right=15] (13.5, -1);

\path[->, color = blue, line width = 0.6mm] (9.75, -1) edge[bend left=10] (4.5, -1);

\draw[cadmium] (6.75,0) -- (6.75, 0.6) -- (7.5, 0.6) -- (7.5, 0.8); 
\node[cadmium] at (7.5, 1.25) {Lemma \ref{lemma:S1}};

\draw[blue] (4.5,0) -- (4.5, 2) -- (6, 2); \draw[blue] (9.75,0) -- (9.75, 2) -- (9, 2);
\node[blue] at (7.5, 2) {Lemma \ref{lemma:S3}};

\draw[red] (1.5,0) -- (1.5, 2.75) -- (6, 2.75); \draw[red] (13.5,0) -- (13.5, 2.75) -- (9, 2.75);
\node[red] at (7.5, 2.75) {Lemma \ref{lemma:S2}};
\end{tikzpicture}
    \caption{\textbf{Managing vertex costs for different values of $c_v$.} For set $V_1$ (the green area), we prove in Lemma~\ref{lemma:S1} that the vertices in this set have no cost. For set $V_2$ (the red area), in Lemma~\ref{lemma:S2} we use the budget of the vertices in $\{v: c_v<1-\tau-\apx\}$ to pay the cost of the vertices in $\{v: c_v> \tau\}$. Finally, for set $V_3$ (the blue area), in Lemma~\ref{lemma:S3} we use the budget of vertices in  $\{v: 0.5< c_v \leq \tau\}$ to pay the cost of the vertices in $\{v: 1-\tau-\apx\leq c_v< 0.5-\apx\}$.
    }
    \label{fig:charging}
\end{figure}

Before stating the three aforementioned lemmas formally, we need the following claim which we will use to prove them.

\begin{claim}\label{claim:uieyuewf}
Consider $\tau$ defined in Algorithm~\ref{alg1}, and let $S$ be the output of Algorithm~\ref{alg3}. For any vertex $v$ with $c_v\in[1-\tau-\apx, \tau]$, we have $\Pr[v\in S]=c_v(2-c_v).$
\end{claim}

\begin{proof}
Consider $M$ and $P$ from Algorithm~\ref{alg3}. Recall that the algorithm outputs $M\cup P$. That is $S=M\cup P$. Set $P$ itself is defined in Algorithm~\ref{alg1} as	
$$P=\{v\in V: c_v>\tau\} \cup \{v\in V: c_v\in[1-\tau-\apx, \tau] \text{ and } v\in C\},$$ 
where $C$ is a minimum vertex cover of a random realization of $G$. This implies that for any $v$ with $c_v\in[1-\tau-\apx, \tau]$, we have 
$$\Pr[v\in S] = \Pr[v\in C\cup M].$$
 Since $M$ and $C$ are minimum vertex covers of two independent realizations of $G$, we get 
$$\Pr[v\in M] = \Pr[v\in C\cup M] = \Pr[v\in C]+ \Pr[v\in M]- \Pr[v\in M]\times \Pr[v\in C] = 2c_v-c_v^2.$$
This concludes the proof of this claim. 	
\end{proof}

\begin{lemma}\label{lemma:S1}
Consider $\tau$ from Algorithm~\ref{alg1}, and define $V_1=\{v: 0.5-\apx\leq c_v\leq 0.5\}$. We have $$\sum_{v\in V_1} b_v - \sum_{v\in V_1} \sigma_v \geq 0,$$ where $b_v$ and $\sigma_v$, the budget and cost of vertex $v$ are  defined in \eqref{eq:budget} and \eqref{eq:cost} with respect to the solution $S$ outputted by Algorithm~\ref{alg3}.
\end{lemma}

\begin{proof}
By Claim~\ref{claim:uieyuewf}, for any vertex $v\in V_1$, we have $\Pr[v\in S]=c_v(2-c_v).$ Combining this with the fact that any vertex in $V_1$ satisfies $c_v>0.5-\apx$, we get 
$$\Pr[v\in S] \leq c_v(2-0.5+\apx) = c_v(1.5+\apx).$$
This means that for any vertex $v\in V_1$, we have 
$b_v = (1.5+\epsilon)c_v  - \Pr[v\in S] \geq 0$ and   $\sigma_v=0,$ and as a result 
$$\sum_{v\in V_1} b_v - \sum_{v\in V_1} \sigma_v \geq 0,$$ completing the proof of this claim.
\end{proof}

\begin{lemma}\label{lemma:S2}
Consider $\tau$ defined in Algorithm~\ref{alg1}, and let
$V_2=\{v: c_v> \tau\} \cup \{v: c_v<1-\tau-\apx\}.$
This set satisfies $$\sum_{v\in V_2} b_v - \sum_{v\in V_2} \sigma_v \geq 0,$$ where $b_v$ and $\sigma_v$, the budget and cost of vertex $v$ are  defined in \eqref{eq:budget} and \eqref{eq:cost} with respect to the solution $S$ outputted by Algorithm~\ref{alg3}.
\end{lemma}

\begin{proof}
Let us start by defining  $$A=\{v: c_v> \tau\} \;\;\;\;\text{ and }\;\;\;\; B=\{v: c_v<1-\tau-\apx\},$$
where $V_2= A\cup B$.
Consider sets $P$ and $M$ from Algorithm~\ref{alg3} which form its output. That is $S= M\cup P$. Note that by definition of $P$ we have $A\subset P$ and as a result $\Pr[v\in S]=1$ for any $v\in A$. Therefore,
$$ b_v - \sigma_v \stackrel{\eqref{eq:cost}, \eqref{eq:budget}}{=}  \max\Big((1.5+\epsilon)c_v- \Pr[v\in S], 0\Big) - \max\Big(\Pr[v\in S] - (1.5+\epsilon)c_v, 0\Big) =  (1.5+\epsilon)c_v-1.$$
Moreover, we have $B\cap P = \emptyset$ and by Observation~\ref{obs:cv}, for any $u\in B$ we have  $\Pr[u\in M] = c_u.$ This implies $\Pr[u\in S] = c_u$ and as a result the budget of  vertex $u$ is
$$b_u= \max\Big((1.5+\epsilon)c_u - \Pr[v\in P\cup M],\, 0\Big) = (1.5+\epsilon)c_u - c_u =  c_u(0.5+\epsilon), $$ and $\sigma_u=0$. Putting these together gives us 
\begin{align*}
    \sum_{v\in V_2} (b_v - \sigma_v) &= \sum_{v\in A} (b_v - \sigma_v)  + \sum_{u\in B} (b_u - \sigma_u) = \sum_{v\in A} \big( (1.5+\epsilon)c_v- 1\big)  + \sum_{u\in B} (0.5+\epsilon)c_u\\
    & = (1.5+\epsilon)\sum_{v\in A} c_v - |A|+ (0.5+\epsilon)\sum_{u\in B} c_u.
\end{align*}
Note that by definition of $\tau$, we have  $\sum_{v\in A} c_v \leq \sum_{u\in B} c_u$. Thus, we can write 

\begin{align*}
    \sum_{v\in V_2} (b_v - \sigma_v) & = (1.5+\epsilon)\sum_{v\in A} c_v - |A|+ (0.5+\epsilon)\sum_{u\in B} c_u \\ 
    &\geq  (2+2\epsilon)\sum_{v\in A} c_v - |A|.
\end{align*}
Since for any vertex $v\in A$ we have $c_v>\tau\geq 0.5$, this implies \begin{align*}
    \sum_{v\in V_2} (b_v - \sigma_v) & \geq  (2+2\epsilon)\sum_{v\in A} c_v - |A| \geq (2+2\epsilon)\sum_{v\in A} 0.5 - |A| \geq \epsilon|A|\geq 0.
\end{align*}
This concludes the proof.
\end{proof}

\begin{lemma} \label{lemma:S3}
Consider $\tau$ defined in Algorithm~\ref{alg1}, and let
$$V_3=\{v: 0.5< c_v \leq \tau\}\cup \{v: 1-\tau-\apx\leq c_v< 0.5-\apx\},$$
This set satisfies $$\sum_{v\in V_3} b_v - \sum_{v\in V_3} \sigma_v \geq 0,$$ where $b_v$ and $\sigma_v$, the budget and cost of vertex $v$ are  defined in \eqref{eq:budget} and \eqref{eq:cost} with respect to the solution $S$ outputted by Algorithm~\ref{alg3}.
\end{lemma}

Due to having a detailed and lengthy proof, we designate Section~\ref{section:proof} to the proof of this lemma.

\noindent Below, we restate our main theorem and give a formal proof for the approximation ratio and the number of queries that our algorithm requires. Later in Section~\ref{section:call}, we explain how we can get the same bounds using only $(n\log n/\epsilon^2)$ calls to the MVC oracle. \\
\smallskip

\noindent\textbf{Theorem~\ref{thm:upper-mvc} (restated).}
\emph{
For any $\epsilon\in (0,0.1)$, Algorithm~\ref{alg1} finds a vertex cover of $\mathcal{G}^\star$ with the expected size of at most $(1.5+\epsilon)\opt$ by querying $O(n/\epsilon p)$ total edges.
}
\smallskip


\begin{proof}
Due to Lemma~\ref{lemma:query}, we know that Algorithm~\ref{alg1}  only requires $O(n/\epsilon p)$ queries.  Let $S$ be the solution outputted by Algorithm~\ref{alg1}. By Observation~\ref{obs:itisavertexcover}, $S$ is a vertex cover of $\mathcal{G}^\star$ and by Observation~\ref{obs:alg1alg3} its expected size is upper-bounded by the output of Algorithm~\ref{alg3}. In Lemma~\ref{lemma:approxratio}, we prove that the output of Algorithm~\ref{alg3} is upper-bounded by $(1.5+\epsilon)\opt.$ Putting these together implies that $S$ is a vertex cover of $\mathcal{G}^\star$ with the expected size of at most  $(1.5+\epsilon)\opt$.
\end{proof}

\subsection{Proof of Lemma~\ref{lemma:S3}}
\label{section:proof}

 Let us define subsets $$A=\{v: 0.5< c_v \leq \tau\} \;\;\;\;\text{ and }\;\;\;\; B=\{v: 1-\tau-\apx\leq c_v< 0.5-\apx\},$$ where $V_3 = A\cup B$. By Claim~\ref{claim:uieyuewf}, for any vertex $v\in V_3$, we have $\Pr[v\in S]=c_v(2-c_v),$ thus
$$\Pr[v\in S]- (1.5+\epsilon)c_v= c_v(2-c_v)- (1.5+\epsilon)c_v=(0.5-\epsilon)c_v - c_v^2.$$ 
Observe that for vertices in $v\in A$ this term is non-positive, therefore this vertex has a zero cost and a budget of
$$b_v = c_v^2-(0.5-\apx)\cdot c_v.$$
On the other hand since $(0.5-\epsilon)c_u - c_u^2>0$ holds for any vertex $u\in B$, this vertex has a zero budget and a cost of
$$\sigma_u =(0.5-\apx)\cdot c_u- c_u^2$$ 

Now, we will show that we can use the budget of vertices in $A$ to pay the cost of the vertices in $B$. Proving that this is possible heavily relies on the way threshold $\tau$ is chosen. 

To complete the proof of this lemma, we need the two following claims.

\begin{claim}\label{claim:bjfjhbr}
For any pair of vertices $u\in B$ and $v\in A$, if  $c_u\geq 1-c_v-\epsilon$, then
 \begin{align*}   b_v - \sigma_{u} \geq (\sigma_{u}/c_{u})  (c_v  - c_{u}). \end{align*}	
\end{claim}

\begin{proof}
 We start by proving  $\sigma_u/c_u\leq b_v/c_v$ as follows.
 \begin{align}
 \sigma_u/c_u &= \frac{(0.5-\apx)\cdot c_u-c_u^2}{c_u} = 0.5-\apx-c_u \nonumber \\
 & \leq 0.5-\apx-(1-c_v-\apx) = c_v-0.5 = \frac{c_v^2-(0.5-\apx)\cdot c_v}{c_v} - \apx \nonumber \\
 &\leq  b_v/c_v-\apx \nonumber\\
 &< b_v/c_v.
 \label{eq:jksetrgnjrejk}	
 \end{align}
Thus, we can write
$$b_v - \sigma_{u} =   c_v (b_v/c_v) - c_{u}(\sigma_{u}/c_{u})\stackrel{\eqref{eq:jksetrgnjrejk}}{\geq} (\sigma_{u}/c_{u})(c_v-c_u),$$
completing the proof of this claim.
\end{proof}

\begin{claim}\label{claim:krjkf}
Let us sort the vertices in $B$ in the increasing order of their $c_u$ with $u_i$ denoting the $i$-th vertex. We claim that for any $i\in |B|$ it is possible to pay the cost of the vertices in $B_i=\{u\in B: c_u \leq c_{u_i}\}$  with the budget of vertices in $A_i=\{v\in A: c_v\geq 1-c_{u_i}-\apx\}$. That is
\begin{align}\label{eq:jkbrnjf3r} \sum_{v\in A_i} b_v  - \sum_{u\in B_i} \sigma_u \geq 0. \end{align}
\end{claim}

\begin{proof}
In order to prove this claim, we prove the following stronger inequality via induction.
\begin{align}\label{eq:jliregkn} \sum_{v\in A_i} b_v  - \sum_{u\in B_i} \sigma_u   \geq (\sigma_{u_i}/c_{u_i}) \left(\sum_{v\in A_i} c_v  - \sum_{u\in B_i} c_u \right).\end{align}
Doing so proves this claim since by definition of $\tau$, for any $i\in |B|$ we have 
\begin{align} \sum_{v\in A_i} c_v  - \sum_{u\in B_i} c_u \geq 0. \label{eq:jbgrhe} \end{align}
As the base case of $i=1$, we need to prove that Equation~\ref{eq:jkbrnjf3r} holds for $A_1=\{v\in A: c_v \geq 1-c_{u_1}-\apx\}$ and $B_1=\{u_1\}$.  That is 
 \begin{align}  \sum_{v\in A_1} b_v - \sigma_{u_1} \geq (\sigma_{u_1}/c_{u_1}) \left(\sum_{v\in A_1} c_v  - c_{u_1} \right).\label{eq:iffjfjk} \end{align}
 
Since for any $v\in A_1$, we have $c_v\geq 1-c_{u_1}-\apx$, this inequality follows from Claim~\ref{claim:bjfjhbr} proving our base case. 
Now, as the induction step, we will prove Equation~\eqref{eq:jkbrnjf3r} for $i=j$ assuming that it holds for $i=j-1$. We can write 
\begin{align*}
\sum_{v\in A_{j}} b_v  - \sum_{u\in B_{j}} \sigma_u &=  \sum_{v\in A_{j-1}} b_v + \sum_{v\in A_j\setminus A_{j-1}}  b_v - \sum_{u\in B_{j-1}} \sigma_{u} - \sigma_{u_{j}} \\
&\geq \left(\sigma_{u_{j-1}}/c_{u_{j-1}}\right) \left(\sum_{v\in A_{j-1}} c_v  - \sum_{u\in B_{j-1}} c_u \right) + \left(\sum_{v\in B_j\setminus B_{j-1}}  b_v - \sigma_{u_{j}}\right)\\
&\geq \left(\sigma_{u_{j}}/c_{u_{j}}\right) \left(\sum_{v\in A_{j-1}} c_v  - \sum_{u\in B_{j-1}} c_u \right) + \left(\sum_{v\in B_j\setminus B_{j-1}}  b_v - \sigma_{u_{j}}\right),
\end{align*}
where the second inequality is due to the induction hypothesis and the last one is due to $$ \sigma_{u_{j-1}}/c_{u_{j-1}} \geq\sigma_{u_{j}}/c_{u_{j}}. $$
To prove the induction step, it suffices to show that the following holds. 

$$\sum_{v\in B_j\setminus B_{j-1}} b_v - \sigma_{u_{j}} \geq\left(\sigma_{u_{j}}/c_{u_{j}}\right)\left(\sum_{v\in B_j\setminus B_{j-1}} c_v - c_{u_{j}} \right).$$
Note that this is a general version of Equation~\eqref{eq:iffjfjk} above. Observe that by definition, for any $u\in B_j\setminus B_{j-1}$  we have $c_v\geq 1-c_{u_j}-\apx$. Thus, by Claim~\ref{claim:bjfjhbr}, we get
 $$\sum_{v\in B_j\setminus B_{j-1}} b_v - \sigma_{u_{j}} \geq\sum_{v\in B_j\setminus B_{j-1}} (\sigma_{u_{j}}/c_{u_{j}})(c_v - c_{u_{j}})  \geq \left(\sigma_{u_{j}}/c_{u_{j}}\right)\left(\sum_{v\in B_j\setminus B_{j-1}} c_v - c_{u_{j}} \right).$$
 This concludes the induction step and the proof of this lemma.
\end{proof}
Observe that proving Claim~\ref{claim:krjkf} also completes the proof of Lemma~\ref{lemma:S3}. Let $x=|B|$. Correctness of Claim~\ref{claim:krjkf} for $i=x$ implies $$\sum_{v\in A} b_v  - \sum_{u\in B} \sigma_u \geq 0.$$
Since vertices in $A$ have a zero cost, this also implies 
$$\sum_{v\in V_3}( b_v - \sigma_u) \geq 0,$$
completing the proof of Lemma~\ref{lemma:S3}.\qed

\section{Working with Approximate $c_v$s}
\label{section:call}

In this section, we discuss how we can implement Algorithm~\ref{alg1} with only polynomial calls to the MVC oracle. We also discuss in Remark~\ref{remark:jkrfj} that this oracle need not be exact. Note that in Algorithm~\ref{alg1}, we need to know $c_v$ of all the vertices (i.e., the probability that a vertex belongs to the minimum vertex cover). We can compute them exactly if we do not limit the number of calls to the oracle. However, we claim that it is possible to use an estimated value of these parameters in Algorithm~\ref{alg1} and still get the same bounds. In the algorithm below we first find an estimate for any  $c_v$ and then feed them to Algorithm~\ref{alg1}.

\vspace{2 mm}
\begin{tboxalg2e}{Oracle-efficient $3/2$-approximation algorithm}
\begin{algorithm}[H]
    \SetKwFunction{SVC}{SVC}
	\DontPrintSemicolon
	\SetAlgoSkip{bigskip}
	\SetAlgoInsideSkip{}
	\label{alg:estimate}
	\SetKwProg{Fn}{Function}{:}{}
 	\DontPrintSemicolon
Draw $t = \frac{n^2}{8 \epsilon^2} \ln\left(2n/\delta \right)$ realizations of $G$ and denote them by $\mathcal{G}_1, \dots, \mathcal{G}_t$. \\
For any $i\in [t]$, let $C_i=\text{MVC}(\mathcal{G}_1).$\\
For any vertex $v\in V$, let $\bar c_v$ be be the fraction of $C_i$'s that contain $v$.\\
Run Algorithm~\ref{alg1}  with parameters $\bar c_v$s and $\epsilon'=\epsilon/2$ (instead of $c_vs$ and $\epsilon$).
\end{algorithm}
\end{tboxalg2e}
\vspace{3 mm}
We first show that Algorithm~\ref{alg:estimate} returns  $\bar{c}_v$s that are within $\epsilon/2n$ of the respective $c_v$s, with high probability. 

\begin{claim} \label{cl:hoeffding}
Setting $\delta=1/n$ in Algorithm~\ref{alg:estimate}, we get $|c_v-\bar{c}_v|\leq \epsilon/2n$ for all $v\in V$ with probability at least $1-1/n$.
\end{claim}
\begin{proof}
This proof follows from a simple application of the Hoeffding bound.
Let $\alpha = \epsilon/2n$.
Let $X^v_i = 1(v\in C_i)$ be an indicator variable for whether $v\in C_i$. Note that $X^v_i$ is a Bernoulli random variable with $\E[X^v_i] = c_v$ for any $i$ and $v\in V$. Furthermore, $\bar c_v = \frac{1}{t}\sum_{i=1}^t X^v_i$.
Using Hoeffding and union bound, we have
\[
\Pr[\exists v\in V \text{ s.t. } |\bar c_v - c_v| \geq \alpha ] \leq 
n \cdot \Pr\left[ \left| \frac{1}{t}\sum_{i=1}^t X^v_i - c_v \right| \geq \alpha  \right] \leq 2n \exp\left( -2 t \alpha^2 \right) \leq \delta = 1/n,
\]
where the last inequality is by the choice of $t = \frac{1}{2\alpha^2} \ln\left(\frac {2n}{\delta} \right)$. 

\end{proof}

Next, we show that Algorithm~\ref{alg:estimate} returns a  vertex cover of $\mathcal{G}^\star$ which w.h.p. has expected size of at most $(3/2+\epsilon)\opt + O(\epsilon)$ and queries only $O(n/p\epsilon)$ edges.

\begin{theorem}\label{claim:nkjehjf}
For any $\epsilon\in (0,0.1)$, Algorithm~\ref{alg:estimate} finds a vertex cover of $\mathcal{G}^\star$ with the expected size of at most $(3/2+\epsilon)\opt + O(\epsilon)$ by querying $O(n/\epsilon p)$ total edges. Moreover, this algorithm is oracle-efficient.
\end{theorem}

\noindent\textit{Proof sketch}.
We will first give an upper-bound on the number of queries (edges in $H$ for Algorithm~\ref{alg1}).
Let $\alpha = \epsilon/2n$. Following the proof of Lemma~\ref{lemma:query}, it is easy to verify that if $e=(u,v)\in H$, then either $\bar{c}_v+\bar{c}_u\leq 1-\epsilon'$ or it is not covered in $C$. By Claim~\ref{cl:hoeffding}, w.h.p., we have $|c_v-\bar{c}_v|\leq \alpha$. Therefore, w.h.p., either $e$ satisfies $c_v+c_u\leq 1-\epsilon'+2\alpha \leq 1-\epsilon$ or it is not covered in $C$. Similar to the proof of Lemma~\ref{lemma:query}, this gives us $|H| = O(n/\epsilon p).$

 We next bound the approximation ratio of Algorithm~\ref{alg:estimate}. Let $\tau$ be the parameter in Algorithm~\ref{alg1} which is run  with parameters $\bar c_v$s and $\epsilon'=\epsilon/2$ (instead of $c_vs$ and $\epsilon$).  That is $\tau$ is the smallest number in $[0.5,1]$ such that $\sum_{v: \bar{c}_v>\tau} \bar{c}_v \leq \sum_{v: \bar{c}_v<1-\tau-\apx'} \bar{c}_v$.  We will partition the vertices to disjoint subsets based on the value of $\tau$.

\begin{itemize}
    \item $A=\{v\in V : \bar c_v>\tau\}.$
    \item $B=\{v\in V : \bar c_v\in [1-\tau-\epsilon,\tau]\}.$
    \item $C= \{v\in V: \bar c_v < 1-\tau-\epsilon\}.$
\end{itemize}
Let $S$ be the vertex cover outputted by this algorithm. By Algorithm~\ref{alg1}, vertices in set $A$ join $S$ with probability one, and vertices in set $C$ join this vertex cover with probability $c_v$ (the probability of $v$ joining \opt). Moreover, following the proof of Claim~\ref{claim:uieyuewf} we can verify that vertices in $B$ join $S$ with probability $2 c_v - c_v^2.$ That is 
$$\E[|S|] = \sum_{v\in A} 1 +  \sum_{v\in B} (c_v-c_v^2) +  \sum_{v\in C} c_v.$$
To prove our desired approximation ratio we will show  $$\E[|S|]\leq (3/2+\epsilon)\opt + O(\epsilon).$$
Recall Equation~\ref{eq:kjfjewjfnjer} which we use in proving Lemma~\ref{lemma:approxratio}:
\begin{align*}\sum_{v\in V} b_v - \sum_{v\in V} \sigma_v \geq 0. \end{align*}
Following our proof steps, one can verify that if Algorithm~\ref{alg1} is run with parameters $\bar c_v$ instead of $c_v$s, this equation implies
\begin{align}\label{eq:jkrfjkhhfbhjer}\sum_{v\in A} 1 + \sum_{v\in B} (2\bar c_v- \bar c_v^2) +  \sum_{v\in C} \bar c_v - (3/2+\epsilon)\opt \leq 0. \end{align}
Since due to Claim~\ref{cl:hoeffding}, for all $v\in V$, we have $|c_v-\bar{c}_v|\leq \epsilon/2n$ with high probability, the following equation also holds with high probability:

\begin{align}\sum_{v\in A} 1 + \sum_{v\in B} (2\bar c_v- \bar c_v^2) +  \sum_{v\in C} \bar c_v - \E[|S|] =  \sum_{v\in B} (2 \bar c_v- \bar c_v^2 - 2c_v +  c_v^2)+  \sum_{v\in C} (c_v-\bar c_v) \geq -2n(\epsilon/2n)\geq -\epsilon. \end{align}
Combining this with \eqref{eq:jkrfjkhhfbhjer} gives us $\E[|S|] - \epsilon - (3/2+\epsilon)\opt \leq 0$ and subsequently
$$\E[|S|] \leq (3/2+\epsilon)\opt  + O(\epsilon).$$ This completes the proof sketch.

\begin{remark}\label{remark:jkrfj}
We remark that the MVC oracle used in the above discussion need not be exact.
Our analysis simply gets an approximation factor relative to the total sum of $c_v$s, which is relative to the approximation power of the oracle’s.
Thus, our results can be seen as obtaining a $3/2  \alpha + \epsilon$ approximation factor given access to an  $\alpha$-approximate oracle.
One benefit of this observation is that it can directly tap into efficient heuristics, such as highly optimized integer programming tools, that are known to work well in practice and achieve highly optimal results, even though provably-efficient approximation algorithms of similar quality does not exist in theory. 
\end{remark}

\section{Tightness Under Mild Correlation}
\label{sec:lower}

In the previous sections, we exhibited an algorithm that gives a $(3/2 + \epsilon)$-approximation for stochastic graphs that have independently realized edges. Indeed, the analysis given in the previous section continues to hold for graphs with a small number of correlated edges. In this section, we show that for such graphs, a $(3/2 + \epsilon)$ approximation factor is tight. That is, we exhibit a stochastic graph with just a few correlated edges and show that any non-adaptive algorithm must have an approximation factor of $(3/2 - \epsilon)$ on this graph with high probability. Our arguments are based on the arguments given in Section 6 of \cite{AKL16}.

\begin{definition}[Mildly Correlated Graph]
\label{def:mildcorr}
We say that an stochastic graph $G = (V,E)$ is mildly correlated if the edge set $E$ can be partitioned into sets $E_1$ and $E_2$ such that the following are satisfied:
\begin{itemize}
    \item The edges in $E_1$ are realized independently from each other: for any $S_1 \subseteq E_1$
    \[\Pr\bigcap_{e \in S_1} \set{e \in G_r} = \prod_{e \in S_1}\Pr\bra{e \in G_r}\]
    \item The edges in $E_1$ are realized independently from those in $E_2$: for any $S_1 \subseteq E_1$ and $S_2 \subseteq E_2$,
    \[\Pr \bra{\bigcap_{e \in S_1} \set{e \in G_r} \enskip \biggr | \enskip \bigcap_{e \in S_2} \set{e \in G_r}} = \prod_{e \in S_1}\Pr\bra{e \in G_r}\]
    \item $E_2$ is small: $|E_2| = O(n)$. 
\end{itemize}
Notably, in our definition of a mildly correlated graph, the realizations of edges in $E_2$ may depend on those in $E_1$: we make no assumptions on probabilities of the form \[\Pr \bra{\bigcap_{e \in S_2} \set{e \in G_r} \enskip \biggr | \enskip \bigcap_{e \in S_1} \set{e \in G_r}},\] where $S_1 \subseteq E_1$ and $S_2 \subseteq E_2$.
\end{definition}

\begin{remark}
Given any mildly correlated stochastic graph $G$, and a parameter $\epsilon\in (0, 0.1)$, Algorithm~\ref{alg1}, outputs a vertex cover of  $\mathcal{G}^\star$ with expected size of at most $(3/2+\epsilon)$  using only $O(n/\epsilon p)$ queries without knowledge of the edge partitions $E_1$ and $E_2$.
\end{remark}

\begin{proof}
 To see this, we first show that Algorithm \ref{alg3} queries at most $O(n/\epsilon p)$ edges. Looking more closely at our argument in Lemma \ref{lemma:query}, our only use of the independence assumption comes from our use of Observation \ref{claim:nkergkjer}. Thus, to show that Lemma \ref{lemma:query} holds in the case of mildly correlated graphs, it suffices to prove that Observation \ref{claim:nkergkjer} holds in this setting as well. Formally, we show that if $\mathcal G$ is a random realization of a mildly correlated stochastic graph $G$, then the number of edges in $G$ not covered by $M$ is at most $O(n/p)$. As before, let $H = G[V\setminus M]$ be the subgraph induced by the complement of the vertex cover $M$. Again, we must have that none of the edges of $H$ are realized in $\mathcal G$ else one of the vertices of $H$ must lie in the vertex cover $M$. Letting $E_1$ and $E_2$ be the subsets of the graph's edge set $E$ as guaranteed by definition \ref{def:mildcorr}, we define $H_1 = H \cap G[E_1]$ and $H_2 = H \cap G[E_2]$. Notice that as $\abs{E_2} = O(n)$, we find that 
 \[\Pr\bra{\abs{H} \ge O(n/p)} = P \bra{\abs{H_1} + \abs{H_2} \ge O(n/p)} \le \Pr\bra{\abs{H_1} \ge O(n/p)} \le (2/e)^n\]
 where in the final inequality, we invoke Observation~\ref{claim:nkergkjer} on the stochastic graph $G[E_1]$, which does have independently realized edges. Thus, $\abs{H}$ does remain small even in this setting, and Lemma \ref{lemma:query} can be applied directly to see that we only sample $O(n/\epsilon p)$ edges. 
 
 Next, we show that the approximation guarantee continues to hold in this setting. The remaining analysis presented in Section \ref{sec:analysis} follows with just a simple change to Algorithm \ref{alg3}: when we sample realizations $\bar{\mathcal H}$ from $\bar H$ in line 4 of the algorithm, we should do so conditioned on the realization $\mathcal H^\star$ that we obtain in line 2. We note that our algorithm (Algorithm \ref{alg1}) does not actually require such capability to function, and we only use Algorithm \ref{alg3} for the purposes of analysis. Once we make this change, we again have that the vertex cover $M$ given in Algorithm \ref{alg3} is drawn independently from the same distribution as the true minimum vertex cover, which is all that is required for the remaining analysis to follow. Thus, Algorithm~\ref{alg1} gives the desired approximation ratio of 1.5 + $\epsilon$ while only sampling $O(n/\epsilon p)$ edges. 
\end{proof}

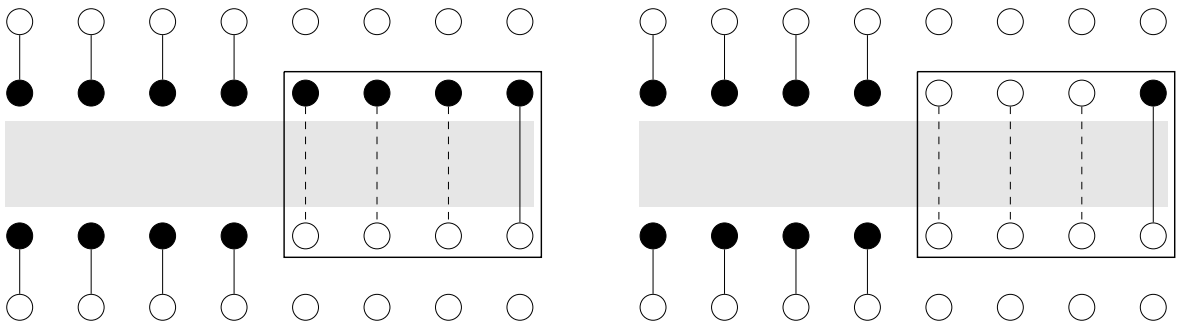
\begin{figure}[h]
    \begin{center}
    \begin{minipage}{0.49\textwidth}
    \centering
    \scalebox{0.95}{\begin{tikzpicture}
  \node[draw, shape=circle] (1tt) at (0 ,1) {};
  \node[draw, shape=circle] (2tt) at (1 ,1) {};
   \node[draw, shape=circle] (3tt) at (2 ,1) {};
  \node[draw, shape=circle] (4tt) at (3 ,1) {};
  \node[draw, shape=circle] (5tt) at (4 ,1) {};
  \node[draw, shape=circle] (6tt) at (5 ,1) {};
   \node[draw, shape=circle] (7tt) at (6 ,1) {};
  \node[draw, shape=circle] (8tt) at (7 ,1) {};

 \node[draw, shape=circle, fill=black] (1t) at (0 ,0) {};
  \node[draw, shape=circle, fill=black] (2t) at (1 ,0) {};
   \node[draw, shape=circle, fill=black] (3t) at (2 ,0) {};
  \node[draw, shape=circle, fill=black] (4t) at (3 ,0) {};
  \node[draw, shape=circle, fill=black] (5t) at (4 ,0) {};
  \node[draw, shape=circle, fill=black] (6t) at (5 ,0) {};
   \node[draw, shape=circle, fill=black] (7t) at (6 ,0) {};
  \node[draw, shape=circle, fill=black] (8t) at (7 ,0) {};
  
  \fill [color=black!10] (-0.2,-0.4) rectangle (7.2,-1.6);
  
   \node[draw, shape=circle, fill=black] (1b) at (0 ,-2) {};
  \node[draw, shape=circle, fill=black] (2b) at (1 ,-2) {};
   \node[draw, shape=circle, fill=black] (3b) at (2 ,-2) {};
  \node[draw, shape=circle, fill=black] (4b) at (3 ,-2) {};
  \node[draw, shape=circle] (5b) at (4 ,-2) {};
  \node[draw, shape=circle] (6b) at (5 ,-2) {};
   \node[draw, shape=circle] (7b) at (6 ,-2) {};
  \node[draw, shape=circle] (8b) at (7 ,-2) {};
  
   \node[draw, shape=circle] (1bb) at (0 ,-3) {};
  \node[draw, shape=circle] (2bb) at (1 ,-3) {};
   \node[draw, shape=circle] (3bb) at (2 ,-3) {};
  \node[draw, shape=circle] (4bb) at (3 ,-3) {};
  \node[draw, shape=circle] (5bb) at (4 ,-3) {};
  \node[draw, shape=circle] (6bb) at (5 ,-3) {};
   \node[draw, shape=circle] (7bb) at (6 ,-3) {};
  \node[draw, shape=circle] (8bb) at (7 ,-3) {};
  
  \draw[line width=0.2mm] (3.7,0.3) -- (7.3, 0.3) -- (7.3, -2.3) -- (3.7, -2.3) -- (3.7,0.3);
  
  \draw (1tt) -- (1t);
  \draw (2tt) -- (2t);
  \draw (3tt) -- (3t);
  \draw (4tt) -- (4t);
  
  
    \draw (1bb) -- (1b);
  \draw (2bb) -- (2b);
  \draw (3bb) -- (3b);
  \draw (4bb) -- (4b);
  
  
      \draw[dashed] (5t) -- (5b);
  \draw[dashed] (6t) -- (6b);
  \draw[dashed] (7t) -- (7b);
  \draw (8t) -- (8b);
\end{tikzpicture}}
    \end{minipage}
    \hfill
    \begin{minipage}{0.49\textwidth}
    \centering
    \scalebox{0.95}{\begin{tikzpicture}
  \node[draw, shape=circle] (1tt) at (0 ,1) {};
  \node[draw, shape=circle] (2tt) at (1 ,1) {};
   \node[draw, shape=circle] (3tt) at (2 ,1) {};
  \node[draw, shape=circle] (4tt) at (3 ,1) {};
  \node[draw, shape=circle] (5tt) at (4 ,1) {};
  \node[draw, shape=circle] (6tt) at (5 ,1) {};
   \node[draw, shape=circle] (7tt) at (6 ,1) {};
  \node[draw, shape=circle] (8tt) at (7 ,1) {};

 \node[draw, shape=circle, fill=black] (1t) at (0 ,0) {};
  \node[draw, shape=circle, fill=black] (2t) at (1 ,0) {};
   \node[draw, shape=circle, fill=black] (3t) at (2 ,0) {};
  \node[draw, shape=circle, fill=black] (4t) at (3 ,0) {};
  \node[draw, shape=circle] (5t) at (4 ,0) {};
  \node[draw, shape=circle] (6t) at (5 ,0) {};
   \node[draw, shape=circle] (7t) at (6 ,0) {};
  \node[draw, shape=circle, fill=black] (8t) at (7 ,0) {};
  
  \fill [color=black!10] (-0.2,-0.4) rectangle (7.2,-1.6);
  
   \node[draw, shape=circle, fill=black] (1b) at (0 ,-2) {};
  \node[draw, shape=circle, fill=black] (2b) at (1 ,-2) {};
   \node[draw, shape=circle, fill=black] (3b) at (2 ,-2) {};
  \node[draw, shape=circle, fill=black] (4b) at (3 ,-2) {};
  \node[draw, shape=circle] (5b) at (4 ,-2) {};
  \node[draw, shape=circle] (6b) at (5 ,-2) {};
   \node[draw, shape=circle] (7b) at (6 ,-2) {};
  \node[draw, shape=circle] (8b) at (7 ,-2) {};
  
   \node[draw, shape=circle] (1bb) at (0 ,-3) {};
  \node[draw, shape=circle] (2bb) at (1 ,-3) {};
   \node[draw, shape=circle] (3bb) at (2 ,-3) {};
  \node[draw, shape=circle] (4bb) at (3 ,-3) {};
  \node[draw, shape=circle] (5bb) at (4 ,-3) {};
  \node[draw, shape=circle] (6bb) at (5 ,-3) {};
   \node[draw, shape=circle] (7bb) at (6 ,-3) {};
  \node[draw, shape=circle] (8bb) at (7 ,-3) {};
  
  \draw[line width=0.2mm] (3.7,0.3) -- (7.3, 0.3) -- (7.3, -2.3) -- (3.7, -2.3) -- (3.7,0.3);
  
  \draw (1tt) -- (1t);
  \draw (2tt) -- (2t);
  \draw (3tt) -- (3t);
  \draw (4tt) -- (4t);
  
  
    \draw (1bb) -- (1b);
  \draw (2bb) -- (2b);
  \draw (3bb) -- (3b);
  \draw (4bb) -- (4b);
  
  
      \draw[dashed] (5t) -- (5b);
  \draw[dashed] (6t) -- (6b);
  \draw[dashed] (7t) -- (7b);
  \draw (8t) -- (8b);
\end{tikzpicture}}
    \end{minipage}
    \end{center}
    \caption{\textbf{A graphical depiction of $G$.} The middle $2n$ vertices (with edges depicted by the gray rectangle) is given by the Ruzsa-Szemer\'edi graph, and its edges are realized independently. We then select $M^*$ (boxed) uniformly at random from the induced matchings, and realize the corresponding exterior edges for all vertices not in $M^*$. With high probability, any algorithm that non-adaptively queries only $O(n)$ edges must cover almost every edge of $M^*$ (as depicted on the left, with vertices in the cover shown in black), but only the few edges in $M^*$ that are actually realized must be covered (as depicted on the right).}
    \label{fig:RS}
\end{figure}

\smallskip

\noindent\textbf{Theorem~\ref{thm:lower-mvc} (restated).}
\emph{
There exists a mildly correlated stochastic graph $G$ for which every non-adaptive algorithm must have an approximation ratio of at least $1.5 - \epsilon$ with probability $1 - o(1)$.  
}

\begin{proof}
 We define $G$ as follows. First, let an $(r, t)$-Ruzsa-Szemer\'edi graph be a bipartite graph on $2n$ vertices whose edge set may be partitioned into $t$ induced matchings of size $r$. Such graphs exist for $r = \frac{n}2 - \epsilon_1$ and $t = n^{\Omega(1/ \log \log n)}$ \cite{DBLP:conf/soda/GoelKK12}. We define the base graph of $G$ by starting with such a graph, and then augmenting it by adding one additional vertex and exterior edge for each of the $2n$ vertices in the Rusza-Szemer\`edi graph. We then realize all of the edges of the Rusza-Szemer\`edi graph independently with $p_e = \epsilon_2$ for all edges $e$ in the edge set. Next, we select one of these induced matchings $M_1, M_2, \dots, M_t$ at random, and call it $M^*$. For each of the $O(n)$ vertices of the Rusza-Szemer\`edi graph that do not participate in $M^*$, we realize its respective exterior edge. It is easy to see that this stochastic graph is mildly correlated, following Definition \ref{def:mildcorr} with $E_1$ denoting the edges of the Ruzsa-Szemer\'edi graph and $E_2$ denoting the exterior edges.
 
 To see that every non-adaptive algorithm must have an approximation ratio of at least $1.5 - \epsilon$, observe that as the algorithm may query at most $O(n) = o\pa{r \cdot t}$ edges, it follows from a simple counting argument that the set of edges that the algorithm queries must contain $o(r)$ edges for all but an $o(1)$ fraction of the matchings $M_1, \dots, M_t$. Thus, with probability $1 - o(1)$, the algorithm must cover $r - o(r)$ edges in $M^*$, and must further cover every exterior edge corresponding to vertices not in $M^*$. As the edges in $M^*$ and the exterior edges do not coincide, it follows that with probability $1 - o(1)$, the size of the vertex cover returned by any non-adapative algorithm must be equal to \[\underbrace{2(n - r)}_{\text{\# exterior edges that must be covered}} + \underbrace{r - o(r)}_{\text{\# edges in $M^*$ that must be covered}} = 1.5n + \epsilon_1 - o(n)\]
 However, observe that by the Chernoff bound, the probability that more than $2n\epsilon_2$ edges in $M^*$ are realized is upper bounded by $o(1)$. Thus, with probability $1 - o(1)$, the size of the minimum vertex cover of $G$ is given by
\[\underbrace{2(n - r)}_{\text{\# exterior edges that must be covered}} + \underbrace{2n\epsilon_2}_{\text{\# edges in $M^*$ that must be covered}} = n(1 + 2\epsilon_2) + \epsilon_1 \]
It thus follows by the union bound that with probability $1 - o(1)$ any such algorithm must have an approximation ratio of
\[\frac{1.5n + \epsilon_1 - o(n)}{n(1 + 2\epsilon_2) + \epsilon_1} \to \frac{1.5}{1 + 2 \epsilon_2} \ge 1.5 - \epsilon\]
for sufficiently large $n$, and appropriate choice of $\epsilon_2$ (note here that this parameter does not depend on $n$, only $\epsilon$). 
\end{proof}

\bibliographystyle{alpha}
\bibliography{refs}
	
\appendix

\end{document}